\documentclass[lettersize,journal]{IEEEtran}
\usepackage{amsmath,amsfonts}
\usepackage{algorithmic}
\usepackage{algorithm}
\usepackage{setspace}
\usepackage{array}
\usepackage{booktabs} 
\usepackage{textcomp}
\usepackage{stfloats}
\usepackage{url}
\usepackage{verbatim}
\usepackage{cite}
\usepackage{amsmath,amssymb,amsfonts}
\usepackage{textcomp}
\usepackage{xcolor}
\usepackage{amsthm}
\usepackage{amsmath}
\usepackage{mathrsfs}
\usepackage{array}
\usepackage{graphicx, subfig}
\usepackage{subfloat}
\usepackage{enumitem}
\usepackage{hyperref}
\usepackage{makecell}  
\usepackage{multirow} 
\usepackage{epstopdf}
\usepackage{soul}
\usepackage{bm}
\theoremstyle{notations}
\DeclareUnicodeCharacter{2212}{\ensuremath{-}}

\newtheorem{Lemma}{Lemma}

\begin{document}
\title{Pinching-Antenna Systems (PASS) Meet Multiple Access: NOMA or OMA? }

\author{
	    Qiao Ren,
	    Xidong Mu,
		Siyu Lin,
		and Yuanwei Liu,~\IEEEmembership{Fellow,~IEEE}
\thanks{Q. Ren and S. Lin are with the School of Electronic and Information Engineering, Beijing Jiaotong University, Beijing 100044, China (e-mail: {qiaoren,sylin}@bjtu.edu.cn).}
\thanks{Xidong Mu is with the Centre for Wireless Innovation (CWI), School of Electronics, Electrical Engineering and Computer Science, Queen's University Belfast, Belfast, BT3 9DT, U.K. (x.mu@qub.ac.uk)}
\thanks{Y. Liu is with the Department of Electrical and Electronic Engineering (EEE), The University of Hong Kong, Hong Kong (e-mail: yuanwei@hku.hk).}
}

\markboth{}%
{Shell \MakeLowercase{\textit{et al.}}: A Sample Article Using IEEEtran.cls for IEEE Journals}
\maketitle

\begin{abstract}
A fundamental two-user PASS-based communication system is considered under three MA schemes, namely non-orthogonal multiple access (NOMA), frequency division multiple access (FDMA), and time division multiple access (TDMA). For each MA scheme, a pinching beamforming optimization problem is formulated to minimize the required transmit power for satisfying users' rate requirements. For NOMA and FDMA, a two-stage algorithm is proposed, where the locations of PAs are derived sequentially by using the successive convex approximation (SCA) method and fine-turning phase adjustment. For TDMA, by leveraging the time-switching feature of PASS, the optimal pinching beamforming of each time slot is derived to maximize the served user channel gain. Numerical results are provided to show that: 1) PASS can achieve a significant performance gain over conventional antenna systems, and 2) NOMA consistently outperforms FDMA, while TDMA provides superior performance than NOMA for symmetric user rate requirements.
\end{abstract}
\begin{IEEEkeywords}
Non-orthogonal multiple access (NOMA), orthogonal multiple access (OMA), pinching-antenna systems (PASS), pinching beamforming. 
\end{IEEEkeywords}

\section{Introduction}
Driven by the further development toward sixth-generation (6G) networks, various flexible-antenna technologies, such as fluid antennas \cite{Wong2021Eluid}, and movable-antenna systems \cite{Zhu2024Modeling}, have recently attracted significant attention for their ability to dynamically adjust the antenna locations, thereby improving wireless communication performance. However, the above flexible-antenna technologies have limitations on mitigating large-scale path loss due to the wavelength-level adjustment range \cite{Ding2025Flexible}. To overcome these shortcomings, pinching-antenna systems (PASS) have been introduced as a revolutionary flexible-antenna technology \cite{Liu2024PASS}. Specially, PASS employ dielectric waveguides as the transmission medium to mitigate long-range path loss, and apply low‐cost dielectric materials, called pinching antennas (PAs), to serve as reconfigurable radiation nodes at any location along the waveguides. This design establishes new line-of-sight links and significantly improves channel quality compared to the conventional antenna technologies. Moreover, PAs can be easily added or removed along the waveguide, enabling highly flexible and dynamic reconfiguration, making PASS particularly well-suited for different communication scenarios \cite{Ding2025Flexible}. 

Since each waveguide must be fed with the same signal stream to PAs, multiple access (MA) techniques are essential for enabling PASS-based multi-user communications over a single waveguide\cite{Ding2025Flexible}. As shown in Fig. \ref{Fig_SystemArchitecture}, MA techniques can be broadly classified into orthogonal multiple access (OMA) and non-orthogonal multiple access (NOMA), depending on whether a given resource block (in time, frequency, or other domains) serves multiple users or only one user. In downlink NOMA, a user with stronger channel gain utilizes successive interference cancellation (SIC) to remove the signals of weaker users before decoding its signal. Typically, the decoding order is determined by the user channel power gains. However, with the aid of PASS, this decoding order can be reconfigured dynamically through PA placement, which expands the available performance trade-offs among users. Although recent studies have investigated integrating NOMA with PASS \cite{Wang2025Antenna, Tegos2025Minimum, Tyrovolas2025Performance}, a comprehensive theoretical comparison between NOMA and OMA in PASS-enabled wireless communications has yet to be thoroughly addressed.
\begin{figure}[!tb] 
	\centering
	\includegraphics[width=0.4\textwidth]{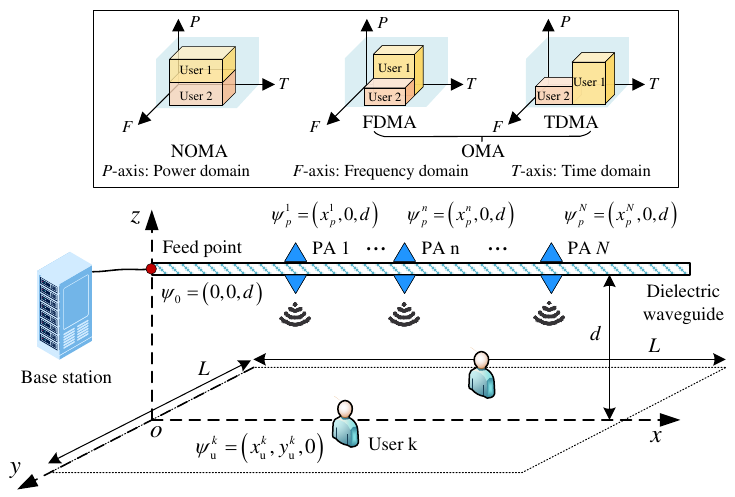} 
	\caption{An illustration of NOMA and OMA in PASS-enabled communications with two users.}
	\label{Fig_SystemArchitecture}
	\vspace{-1.5em}
\end{figure}

In this paper, we study a two-user PASS-based communication system under three MA schemes, including NOMA, frequency division multiple access (FDMA), and time division multiple access (TDMA). For each MA scheme, we formulate a pinching beamforming optimization problem that minimizes the required transmit power while satisfying rate requirements of the users. For NOMA and FDMA, we develop a two-stage algorithm to tackle these non-convex problems, where the coarse locations of PAs are derived by the successive convex approximation (SCA) method. Given the initial positioning, we then fine-tune the PA positions to ensure constructive signal superposition at users and derive a high-quality pinching beamforming solution. For TDMA, with the time-switching feature of PASS, we can derive the pinching beamforming of each time slot to maximize the served user channel gain. Numerical results reveal that 1) PASS can achieve a higher performance gain than conventional antenna systems, and 2) NOMA consistently achieves a better performance than FDMA, while TDMA outperforms NOMA for symmetric user rate requirements due to the time-switching feature of PASS.

\section{System Model and Problem Formulation}
We consider a fundamental single-waveguide PASS-based two-user communication system, consisting of one access point (AP), $2$ single-antenna users, and a dielectric waveguide equipped with $N$ PAs, as illustrated in Fig. \ref{Fig_SystemArchitecture}. The two users, indexed by $k\in\left\{1,2\right\}$, are randomly distributed within a square area of side length $L$ meters. Fig. \ref{Fig_SystemArchitecture} illustrates two OMA schemes: FDMA and TDMA, which serve two users in two orthogonal frequency or time resources, respectively. For NOMA, the same time/frequency resources are shared for two users. The location of user $k$ is denoted by $\psi_{\mathrm{u}}^k=\left[x_{\mathrm{u}}^k, y_{\mathrm{u}}^k, 0\right]$. We assume that the waveguide with the length of $L$ is deployed parallel to the $x$-axis at height $d$, with the feed point at $\psi_{0}=\left[0, 0, d\right]$. Each PA $n$ is placed at $\psi_{\mathrm{p}}^n=\left[x_{\mathrm{p}}^n, 0, d\right]$, and the vector of $x$-axes is $\mathbf{x}_{\mathrm{p}}=\left[x_{\mathrm{p}}^{1}, \ldots, x_{\mathrm{p}}^{N}\right]$. We make the ideal assumption that the PAs can be activated continuously along the waveguide over a distance much greater than the wavelength, with the minimum spacing ensured to prevent antenna coupling. Then, the feasible set of PAs locations is given by $\mathcal{F}=\left\{x_{\mathrm{p}}^n \mid 0 \leq x_{\mathrm{p}}^n \leq L, x_{\mathrm{p}}^{n+1}-x_{\mathrm{p}}^{n} \geq \Delta\right\}, \forall n=\{1,\dots,N\} $. 

The in-waveguide channel vector between the feed point and the PAs can be expressed as 
\begin{align} 
	\mathbf{g}\left(\mathbf{x}_{\mathrm{p}}\right) &= \frac{1}{\sqrt{N}}{\left[ e^{-j \frac{2 \pi}{\lambda_{\mathrm{g}}}\left\| \psi_{\mathrm{p}}^1 - \psi_0\right\|}, \ldots, e^{-j \frac{2 \pi}{\lambda_{\mathrm{g}}}\left\| \psi_{\mathrm{p}}^N - \psi_0\right\|} \right]^T},
\end{align}
where $\lambda_{\mathrm{g}}=\frac{\lambda}{n_{\mathrm{eff}}}$ is the guided wavelength, $\lambda$ is the signal wavelength in the free space, and $n_{\mathrm{eff}}$ is the effective refractive index of the waveguide \cite{Tegos2025Minimum}. We assume that the total transmit power is equally distributed among the $N$ PAs. The wireless channel vector between the PAs and the user $k$ can be expressed as 
\begin{equation}
\mathbf{h}_{k}\left(\mathrm{\mathbf{x}}_{\mathrm{p}}\right)
= {\left[\frac{\eta e^{-j \frac{2 \pi}{\lambda}\left\| \psi_{\mathrm{u}}^k - \psi_{\mathrm{p}}^1\right\|}}{\left\|\psi_{\mathrm{u}}^k - \psi_{\mathrm{p}}^1\right\|}, \ldots, \frac{\eta e^{-j \frac{2 \pi}{\lambda}\left\| \psi_{\mathrm{u}}^k - \psi_{\mathrm{p}}^N\right\|}}{\left\|\psi_{\mathrm{u}}^k - \psi_{\mathrm{p}}^N\right\|} \right]^T},
\end{equation}
where $\eta=\frac{c}{4 \pi f_c}$, with the speed of light $c$, and the carrier frequency $f_c$. The Euclidean distance $\left\| \psi_{\mathrm{u}}^k - \psi_{\mathrm{p}}^n\right\|$ between user $k$ and the PA $n$ can be expressed as \cite{Zhu2025Secure}
\begin{equation}
\left\| \psi_{\mathrm{u}}^k - \psi_{\mathrm{p}}^n\right\| = \sqrt{\left(x_{\mathrm{u}}^k-x_{\mathrm{p}}^n\right)^2+ D_k^2}.
\end{equation}
where $D_k = \sqrt{\left(y_{\mathrm{u}}^k\right)^2 + d^2}$.

\subsection{NOMA}
For NOMA, the signals of users are transmitted by utilizing the superposition coding, which is given by
\begin{equation}
s=\sqrt{P_1}s_1 + \sqrt{P_2}s_2,
\end{equation}
where $P_k$ and $s_k$ denote the transmit power and the normalized transmitted data symbol for user $k$, respectively, with $\mathbb{E}\left[s_k s_k^*\right]=1$. Accordingly, the superimposed signal received at the user $k$ can be written as
\begin{equation}
		y_k = \mathbf{h}_{k}^T\left(\mathrm{\mathbf{x}}_{\mathrm{p}}\right)\mathbf{g}\left(\mathbf{x}_{\mathrm{p}}\right)\left(\sqrt{P_1}s_1 + \sqrt{P_2}s_2\right) + n_k,
\end{equation}
where $n_k \sim \mathcal{C N}\left(0, \sigma_k^2\right)$ is the additive white Gaussian noise at the user $k$, and $\sigma_k^2$ is the noise power. SIC is employed to avoid multiple-access interference. Specifically, the user with stronger channel first decodes the signal intended for the weaker user, removes it from its received signal, and then decodes its data. Thus the achievable data rate of the user $k \in \left\{1,2\right\}$ is given by 
\begin{equation}
	R_k^\mathrm{N} =\log_2\left(1+\frac{P_k |v_k\left(\mathrm{\mathbf{x}}_{\mathrm{p}}\right)|^2}{\lambda_k P_{\bar{k}} |v_k\left(\mathrm{\mathbf{x}}_{\mathrm{p}}\right)|^2 +N \sigma_k^2}\right),
\end{equation}
where $|v_k\left(\mathrm{\mathbf{x}}_{\mathrm{p}}\right)|^2$ is the channel power gain as follows
\begin{subequations}\label{Eq_v_k}
	\begin{align}
		\left|v_k\left(\mathrm{\mathbf{x}}_{\mathrm{p}}\right)\right|^2 &= \left|\sum_{n=1}^N \frac{\eta e^{-j \phi_k^n}}{\left\|{\psi}_{\mathrm{u}}^k-{\psi}_{\mathrm{p}}^{n}\right\|}\right|^2 \nonumber\\
		&= \left| \sum_{n=1}^N \frac{\eta e^{-j \left[\frac{2 \pi}{\lambda}\left\|\psi_{\mathrm{u}}^k-\psi_{\mathrm{p}}^{n}\right\|+\frac{2 \pi}{\lambda_g} \left\|\psi_{\mathrm{p}}^n-\psi_0\right\| \right]}}{\left\|{\psi}_{\mathrm{u}}^k-{\psi}_{\mathrm{p}}^{n}\right\|}\right|^2. \tag{\ref{Eq_v_k}}
	\end{align}
\end{subequations}
Note that $\bar{k}=1$, if $k=2$; and $\bar{k}=2$, otherwise. With two users, there are $2!=2$ possible SIC decoding orders. Therefore, the binary variables, $\lambda_k \in \left\{0,1\right\}, k \in \left\{1,2\right\}$, are introduced to denote the decoding orders and $\lambda_1 + \lambda_2 = 1$. In particular, if $\left|v_1\right|^2 \geq  \left|v_2\right|^2$ (i.e., user 1 is the stronger user), we set $\lambda_1=0$ and $\lambda_2=1$; otherwise, $\lambda_1=1$ and $\lambda_2=0$. We aim to minimize the transmit power at AP by optimizing PA locations $\mathbf{x}_{\mathrm{p}}$, subject to rate requirements of users and continuous antenna activation constraints. The optimization problem is formulated as
\begin{subequations}\label{Pro_NOMA}
	\begin{align}
			\min _{\lambda_k, \mathbf{x}_{\mathrm{p}}, P_1, P_2} & P_1+P_2 \\ 
			\text {s.t.}~~
			&R_k^\mathrm{N} \geq \gamma_k, \forall k, \\
			&0 \leq x_{\mathrm{p}}^n \leq L, \forall n, \\
			& x_{\mathrm{p}}^{n+1}-x_{\mathrm{p}}^{n} \geq \Delta, \forall n, \\
			&\lambda_k \in \left\{0,1\right\},  \sum\nolimits_{k=1}^2 \lambda_k =1, 
	\end{align}
\end{subequations}
where $\gamma_k$ is the minimum rate requirement for user $k$. Since the achievable rate $R_k^\mathrm{N}$ is monotonically increasing in $P_k$, the inequality rate constraint (\ref{Pro_NOMA}a) must be satisfied with equality at the optimal solution. Although there are two possible SIC decoding orders, both of them yield similar characteristics. Hence, the optimal solution is determined by selecting the minimum transmit power between the two cases by applying exhaustively searching. For clarity of presentation, we focus on the specific case $\lambda_1=0$ and $\lambda_2=1$. Under this decoding order, the optimal transmit powers $P_1$ and $P_2$ can be expressed as
\begin{equation}
	\left\{
	\begin{aligned}
		& P_1=\frac{\left(2^{\gamma_1}-1\right) N \sigma_1^2}{|v_1\left(\mathrm{\mathbf{x}}_{\mathrm{p}}\right)|^2}, \\
		& P_2=\frac{\left(2^{\gamma_1}-1\right)\left(2^{\gamma_2}-1\right) N \sigma_1^2}{|v_1\left(\mathrm{\mathbf{x}}_{\mathrm{p}}\right)|^2} + \frac{\left(2^{\gamma_2}-1\right) N \sigma_2^2}{|v_2\left(\mathrm{\mathbf{x}}_{\mathrm{p}}\right)|^2}.
	\end{aligned}
	\right.
\end{equation}
The corresponding optimization problem is transformed into
\begin{subequations}\label{Pro_NOMA_simp}
	\begin{align}
		\min _{\mathbf{x}_{\mathrm{p}}}\quad &  \frac{\left(2^{\gamma_1}-1\right) 2^{\gamma_2} N \sigma_1^2}{|v_1\left(\mathrm{\mathbf{x}}_{\mathrm{p}}\right)|^2} + \frac{\left(2^{\gamma_2}-1\right) N \sigma_2^2}{|v_2\left(\mathrm{\mathbf{x}}_{\mathrm{p}}\right)|^2} \\ 
		\text {s.t.}~~
		& (\ref{Pro_NOMA}c), (\ref{Pro_NOMA}d).
	\end{align}
\end{subequations}

\subsection{OMA}
\subsubsection{FDMA}
For FDMA, the AP communicates with the users simultaneously in two orthogonal frequency slots of equal size. The achievable data rate of the user $k \in \left\{1,2\right\}$ can be expressed as
\begin{equation}
R_k^\mathrm{F}=\frac{1}{2} \log_2\left(1+\frac{P_k |v_k\left(\mathrm{\mathbf{x}}_{\mathrm{p}}\right)|^2}{\frac{1}{2}N\sigma_k^2}\right). 
\end{equation}
Similar to the NOMA scheme, the relevant inequality rate constraint can be transformed into equality constraints. Accordingly, the power minimization problem for FDMA can be formulated as
\begin{subequations}\label{Pro_FDMA_simp}
	\begin{align}
		\min _{\mathbf{x}_{\mathrm{p}}} \quad &  \frac{\left(2^{2\gamma_1}-1\right) N \sigma_1^2}{2 |v_1\left(\mathrm{\mathbf{x}}_{\mathrm{p}}\right)|^2} + \frac{\left(2^{2\gamma_2}-1\right) N \sigma_2^2}{2|v_2\left(\mathrm{\mathbf{x}}_{\mathrm{p}}\right)|^2}\\ 
		\text {s.t.}~~
		& (\ref{Pro_NOMA}c), (\ref{Pro_NOMA}d). 
	\end{align}
\end{subequations}

\subsubsection{TDMA}
For TDMA, the AP communicates with the users simultaneously in two orthogonal time slots of equal size. Note that PASS enable each time slot to have its own PA activation configuration under TDMA, i.e., a time-switching feature. Unlike NOMA and FDMA schemes, where an identical PA activation vector $\mathrm{\mathbf{x}}_{\mathrm{p}}$ is employed to both users, the PA deployment for TDMA is different due to the time-switching feature of PASS. Therefore, the achievable rate for user $k \in \left\{1,2\right\}$ is given by
\begin{equation}\label{Eq_R_kT}
	R_k^\mathrm{T}=\frac{1}{2}\log_2\left(1+\frac{2 P_k |v_k\left(\mathrm{\mathbf{x}}_{\mathrm{p},k}\right)|^2}{N\sigma_k^2}\right),
\end{equation}
where $\mathrm{\mathbf{x}}_{\mathrm{p,k}}=[x_{\mathrm{p,k}}^{1}, \ldots, x_{\mathrm{p,k}}^{N}]$ denotes the PA location vector for user $k$. An additional factor of 2 is applied to the received power term to ensure a fair comparison with NOMA. The power minimization problem for TDMA can be formulated as
\begin{subequations}\label{Pro_TDMA_simp}
	\begin{align}
		\min _{\mathbf{x}_{\mathrm{p,1}}} \quad &  \frac{\left(2^{2 \gamma_1}-1\right) N\sigma_1^2}{2 |v_1\left(\mathrm{\mathbf{x}}_{\mathrm{p},1}\right)|^2} + \min _{\mathbf{x}_{\mathrm{p,2}}} \frac{\left(2^{2 \gamma_2}-1\right) N\sigma_2^2}{2 |v_2\left(\mathrm{\mathbf{x}}_{\mathrm{p},2}\right)|^2} \\ 
		\text {s.t.}~~
		& (\ref{Pro_NOMA}c), (\ref{Pro_NOMA}d).
	\end{align}
\end{subequations}

\section{Proposed Solutions}\label{Sec_Solution}
\subsection{NOMA and FDMA}
Note that the channel power gain of users in the objective functions of (\ref{Pro_NOMA_simp}) for NOMA, and (\ref{Pro_FDMA_simp}) for FDMA comprises both large-scale path loss and the phase alignment among the PAs. Since the free-space and waveguide wavelengths are significantly smaller than the distance between PAs and users, even slight PA position adjustment can induce a periodic phase variation $\phi_k^n,\forall k \in \{1,2\}$ in (\ref{Eq_v_k}). By precisely tuning the PA positions, the phase $\phi_k^n$ can be coherently aligned and enable constructive combining of the signal, and thus improve the channel power gain at user $k$. The corresponding phase alignment constraint can be expressed as
\begin{equation}\label{phi}
	\phi_k^n-\phi_k^{n-1}=2 m_k \pi, \forall m_k \in \mathcal{Z}, n \in \mathcal{N}, k \in \{1,2\}.
\end{equation}

Leveraging this decoupling property, we can first tackle the large‐scale path‐loss optimization. Since NOMA and FDMA share a similar problem structure, we focus on NOMA scheme for clarity. The problem for NOMA can be relaxed as
\begin{subequations}\label{Pro_NOMA_relax}
	\begin{align}
		\min _{\mathbf{x}_{\mathrm{p}}}\quad & \frac{d_1}{ | \sum_{n=1}^N \frac{1}{ \left\|{\psi}_{\mathrm{u}}^1-{\psi}_{\mathrm{p}}^{n}\right\|} |^2} + \frac{d_2}{ | \sum_{n=1}^N \frac{1}{ \left\|{\psi}_{\mathrm{u}}^2-{\psi}_{\mathrm{p}}^{n}\right\|} |^2}\\ 
		\text {s.t.}~~
		& (\ref{Pro_NOMA}c), (\ref{Pro_NOMA}d), (\ref{phi}),
	\end{align}
\end{subequations}
where $d_1 = \frac{\left(2^{\gamma_1}-1\right) 2^{\gamma_2} N \sigma_1^2}{\eta^2} $, and $d_2 = \frac{\left(2^{\gamma_2}-1\right) N \sigma_2^2}{\eta^2}$. Since PA locations impact the large-scale loss and phase-shift constraints at different scales, we propose a two-stage algorithm in which the SCA method is first used to determine the overall PA distribution by minimizing the objective function without considering (\ref{phi}), followed by fine-tuning PA position to satisfy the constructive phase alignment constraints for both users.

\subsubsection{Coarse PA locations Design}
Without considering the constraint (\ref{phi}), the coarse PA locations subproblem (\ref{Pro_NOMA_relax}) can be transformed as
\begin{subequations}\label{Pro_NOMA_relax2}
	\begin{align}
	\min _{\mathbf{x}_{\mathrm{p}}}\quad & \frac{d_1}{ | \sum_{n=1}^N \frac{1}{ \sqrt{(x_{\mathrm{u}}^1-x_{\mathrm{p}}^n)^2+ D_1^2}} |^2} + \frac{d_2}{ | \sum_{n=1}^N \frac{1}{ \sqrt{(x_{\mathrm{u}}^2-x_{\mathrm{p}}^n)^2+ D_2^2}} |^2}\\ 
	\text {s.t.}~~
	& (\ref{Pro_NOMA}b), (\ref{Pro_NOMA}c),
\end{align}
\end{subequations}
To tackle this non-convex problem, we introduce auxiliary variables $\theta^n_k>0$, and let
\begin{equation}\label{Eq_theta_nk1}
\frac{1}{\theta^n_k} \geq \sqrt{(x_{\mathrm{u}}^k-x_{\mathrm{p}}^n)^2+D_k^2}, \quad \forall k \in \left\{1,2\right\}.
\end{equation}
which is jointly convex on both sides w.r.t. $x_{\mathrm{u}}^k$, and $\mathbf{\theta}$, respectively. Then, we employ SCA method by linearizing $\frac{1}{\theta^n_k}$ via the first-order Taylor expansion at $\theta^{n,(l)}_{k}$ as follows
\begin{equation}\label{Eq_theta_nk3}
	\frac{1}{\theta^n_k} \geq f(\theta^n_k | \theta^{n,(l)}_{k}) \triangleq \frac{1}{\theta^{n,(l)}_{k}} - \frac{1}{(\theta^{n,(l)}_{k})^2} \left(\theta^n_k - \theta^{n,(l)}_{k}\right)
\end{equation}
As a result, for given $\theta^{n,(l)}_{k}$, by replacing $\frac{1}{\theta^n_k}$ with its linear lower bound, problem (\ref{Pro_NOMA_relax2}) can be transformed as
\begin{subequations}\label{Pro_NOMA_final}
	\begin{align}
		\min _{\mathbf{x}_{\mathrm{p}},\boldsymbol{\theta}_1,\boldsymbol{\theta}_2 }\quad & \frac{d_1}{ | \sum_{n=1}^N \theta^n_1 |^2} + \frac{d_2}{ | \sum_{n=1}^N \theta^n_2 |^2} \\
		\text {s.t.}~~
		& f(\theta^n_k | \theta_n^{k,(l)}) \geq \sqrt{(x_{\mathrm{u}}^k-x_{\mathrm{p}}^n)^2+D_k^2}, \\
		& \theta^n_1, \theta^n_2 > 0, \\
		& (\ref{Pro_NOMA}c), (\ref{Pro_NOMA}d),
	\end{align}
\end{subequations}
where $\boldsymbol{\theta}_1=\left[\theta_1^1, \ldots, \theta_1^N\right]$, and $\boldsymbol{\theta}_2=\left[\theta_2^1, \ldots, \theta_2^N\right]$. Problem (\ref{Pro_NOMA_final}) is a convex optimization problem, which can be solved by standard convex problem solvers such as CVX \cite{Stephen2004Convex}.

\subsubsection{Fine-Tuning Phase-Coherent Beamforming Design} 
With the optimal coarse PA positions, we can refine their locations to satisfy the phase alignment constraint (\ref{phi}) by moving PAs on the wavelength scale, thereby ensuring that signals sent by different PAs are constructively combined at each user. In this case, the phase $\phi_k^n$ can be approximated around its neighboring PA phase $\phi_k^{(n-1)}$ by first-order Taylor expansion, which is given by \cite{Zhu2025Secure}
\begin{subequations}\label{phi_k}
	\begin{align}
		\phi_k^n &\approx \frac{2\pi}{\lambda} \frac{x_{\mathrm{p}}^{n-1} - x_{\mathrm{u}}^k}{\sqrt{(x_{\mathrm{p}}^{n-1} - x_{\mathrm{u}}^k)^2 + D_k^2}} (x_{\mathrm{p}}^n - x_{\mathrm{p}}^{n-1}) \tag{\ref{phi_k}}\\
		&+ \frac{2\pi}{\lambda} \sqrt{(x_{\mathrm{p}}^{n-1} - x_{\mathrm{u}}^k)^2 + D_k^2}  + \frac{2\pi}{\lambda_g} x_{\mathrm{p}}^{n} \nonumber.
	\end{align}
\end{subequations}
Hence, for each user $k$, the constraint (\ref{phi}) can be equivalently reformulated as
\begin{equation}\label{Eq_Deltax_equation}
	\Delta x \left[ \frac{1}{\lambda} \frac{x_{\mathrm{p}}^{n-1} - x_{\mathrm{u}}^k}{\sqrt{(x_{\mathrm{p}}^{n-1} - x_{\mathrm{u}}^k)^2 + D_k^2}} + \frac{1}{\lambda_g} \right] \!\!= m_k, \forall k \in \{1,2\},
\end{equation}
where $\Delta x = x_{\mathrm{p}}^{n}- x_{\mathrm{p}}^{n-1}$. Combining the equations (\ref{Eq_Deltax_equation}) for both users, we obtain the following expression for the initial phase‐adjustment $\widehat{\Delta} x$ as follows
\begin{equation}\label{Eq_widehatDeltaRight}
	\widehat{\Delta} x = \frac{\lambda(m_1 - m_2)}{\frac{x_{\mathrm{p}}^{n-1} - x_{\mathrm{u}}^1}{\sqrt{(x_{\mathrm{p}}^{n-1} - x_{\mathrm{u}}^1)^2 + D_1^2}} - \frac{x_{\mathrm{p}}^{n-1} -  x_{\mathrm{u}}^2}{\sqrt{(x_{\mathrm{p}}^{n-1} - x_{\mathrm{u}}^2)^2 + D_2^2}}}
\end{equation}
Therefore, given the position of the ($n−1$)-th PA, we can refine the position of $n$-th PA as ${x}_{\mathrm{p}}^{n} = x_{\mathrm{p}}^{n-1} + \widehat{\Delta} x$. Similarly, for a given $x_{\mathrm{p}}^{n}$, we can refine the position of $n-1$-th PA as ${x}_{\mathrm{p}}^{n-1} = x_{\mathrm{p}}^{n} - \widehat{\Delta} x^{\prime}$, with $\widehat{\Delta} x^{\prime}$ given by  
\begin{equation}\label{Eq_widehatDeltaLeft}
\widehat{\Delta} x^{\prime} = \frac{\lambda(m_1 - m_2)}{\frac{x_{\mathrm{p}}^{n} - x_{\mathrm{u}}^1}{\sqrt{(x_{\mathrm{p}}^{n} - x_{\mathrm{u}}^1)^2 + D_1^2}} - \frac{x_{\mathrm{p}}^{n} -  x_{\mathrm{u}}^2}{\sqrt{(x_{\mathrm{p}}^{n} - x_{\mathrm{u}}^2)^2 + D_2^2}}}.
\end{equation}

Because the above combination step removes all $\lambda_g$ terms, the current phase $\widehat{\Delta} x$ and $\widehat{\Delta} x^{\prime}$ no longer satisfy the original phase‐alignment constraints (\ref{Eq_Deltax_equation}) for both users. Hence, ${x}_{\mathrm{p}}^{n}$ must be further refined to achieve a feasible solution. Actually, a single uniform PA phase adjustment $ \Delta x$ cannot maximize channel gains for both users simultaneously in general. Rather than enforcing (\ref{Eq_Deltax_equation}) exactly for both users, we seek a balance to minimize transmit power via the following fine‐tuning procedure. To preserve the large-scale path loss determined in the first stage, we use the central PA as the reference point, index $c=(\frac{N}{2})$, for all subsequent adjustments.
\begin{itemize}
	\item According to coarse PA locations design, the obtained location on $x$-axes of the $(\frac{N}{2})$-th PA is $x_{\mathrm{p}}^{c}$. 
	\item For $n = \frac{N}{2}+1,\dots,N$, we first derive the relative offset $\delta = x_{\mathrm{p}}^{n} - x_{\mathrm{p}}^{c}$. 
	\item Next, if $\delta \!\!\!\!\! \geq \!\!\!\!\! 4\Delta$, to satisfy the feasible set of PA location constraints, we set the segment between $\left[x_{\mathrm{p}}^{n}-3\Delta, x_{\mathrm{p}}^{n}+3\Delta\right]$. Otherwise, we restrict the segment between $\left[x_{\mathrm{p}}^{c}+\Delta, x_{\mathrm{p}}^{n}+3\Delta\right]$. Within the segment, we choose the refined location on $x$-axes which minimizes the objective in the problem (\ref{Pro_NOMA_final}). 
	\item Then, we apply this refinement sequentially for the location of $n$-th ($n= \frac{N}{2}+1,\dots, N$) PA.
	\item Similarly, we repeat the above steps in reverse order for $n$-th ($n=\frac{N}{2}-1,\dots,1$). 
\end{itemize}
The desired PA locations ${x}_{\mathrm{p}}^{n,*}$ can be obtained by applying the above fine‐tuning procedure sequentially across all antennas. The same two‐stage algorithm can be applied to solve (\ref{Pro_FDMA_simp}) for FDMA. With the interior‐point method, the computational complexity of the SCA algorithm is about $\mathcal{O}(I_{\mathrm{iter}}N^{3.5})$, where $I_{\mathrm{iter}}$ is the total number of SCA iterations. The computational complexity for fine-turning design is $\mathcal{O}(N)$, which is negligible in comparison. Hence, the overall computational complexity for the proposed algorithm is $\mathcal{O}(I_{\mathrm{iter}}N^{3.5})$.

\subsection{TDMA}
For TDMA, the two PA deployment vectors $\mathbf{x}_{\mathrm{p,1}}$ and ${\mathbf{x}_{\mathrm{p,2}}}$ are decoupled in the objective function of (\ref{Pro_TDMA_simp}) due to time-switching of PASS. Consequently, the power minimization problem for user $k$ can be simplified as
\begin{subequations}\label{Pro_TDMA_simp2}
	\begin{align}
		\min _{\mathbf{x}_{\mathrm{p,k}}} \quad &  \frac{\left(2^{2 \gamma_1}-1\right) N\sigma_1^2}{2 |v_1\left(\mathrm{\mathbf{x}}_{\mathrm{p},k}\right)|^2} \\ 
		\text {s.t.}~~
		& (\ref{Pro_NOMA}c), (\ref{Pro_NOMA}d).
	\end{align}
\end{subequations}
To solve this problem, the global optimal solution is provided in the following lemma.
\begin{Lemma}
The optimal solution $\mathbf{x}_{\mathrm{p,k}}$ deploys the $N$ PAs symmetrically around $x_{\mathrm{u}}^k$ with equal spacing $\Delta$, provided that $D_k \gg \Delta$, which is given by
\begin{equation}
	x_{\mathrm{p,k}}^*=x_{\mathrm{u}}^k-\frac{(N-1)}{2} \Delta+(i-1) \Delta, i=1, \cdots, N.
\end{equation}
\end{Lemma}
\begin{proof}
See \cite[Appendix B]{Xu2025RateMax}.
\end{proof}
By applying an additional phase refinement step, the optimal PA locations under TDMA can be fully determined \cite{Xu2025RateMax}.

\begin{figure*}[htbp]
	\centering
	\begin{minipage}[t]{0.32\textwidth}
		\raggedleft
		\includegraphics[width=\textwidth]{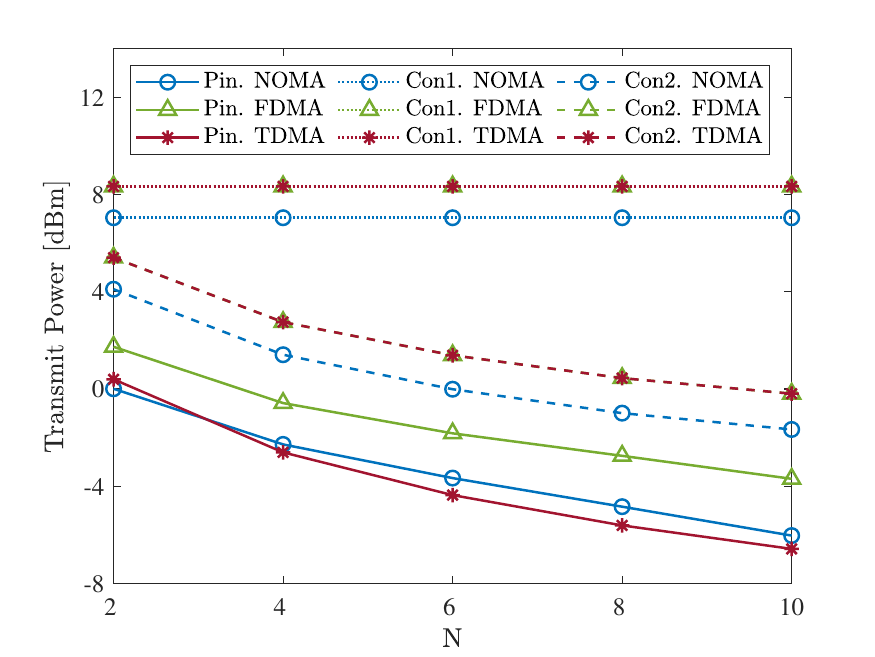}
		\caption{The transmit power changes with the number of PAs $N$.}
		\label{Fig:Fig_NvsP}
	\end{minipage}\hfill
	\begin{minipage}[t]{0.32\textwidth}
		\centering
		\includegraphics[width=\textwidth]{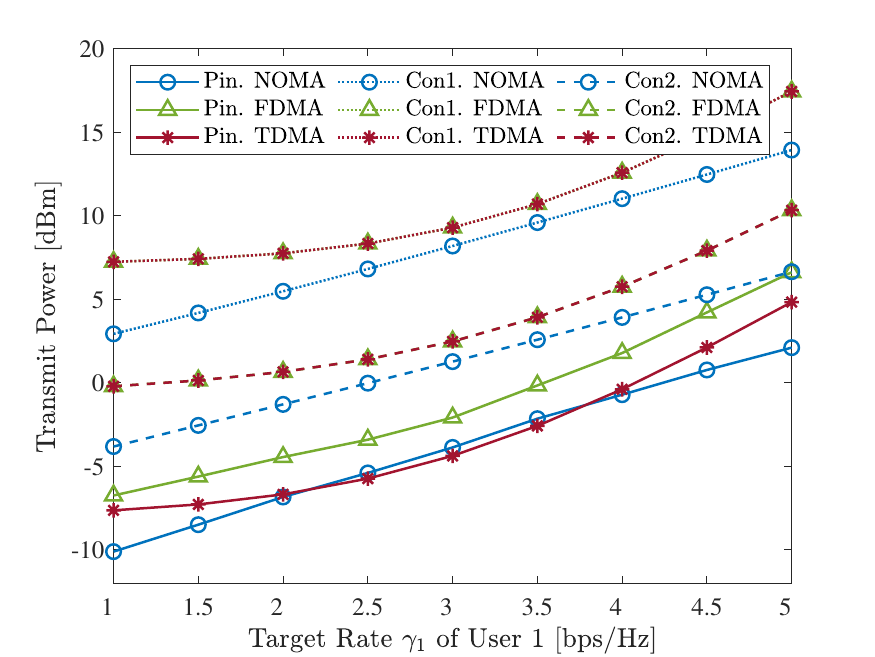}
		\caption{The transmit power changes with the target rate $\gamma_1$ of user 1. }
		\label{Fig:Fig_RatevsP}
	\end{minipage}\hfill
	\begin{minipage}[t]{0.32\textwidth}
		\raggedright
		\includegraphics[width=5.67cm]{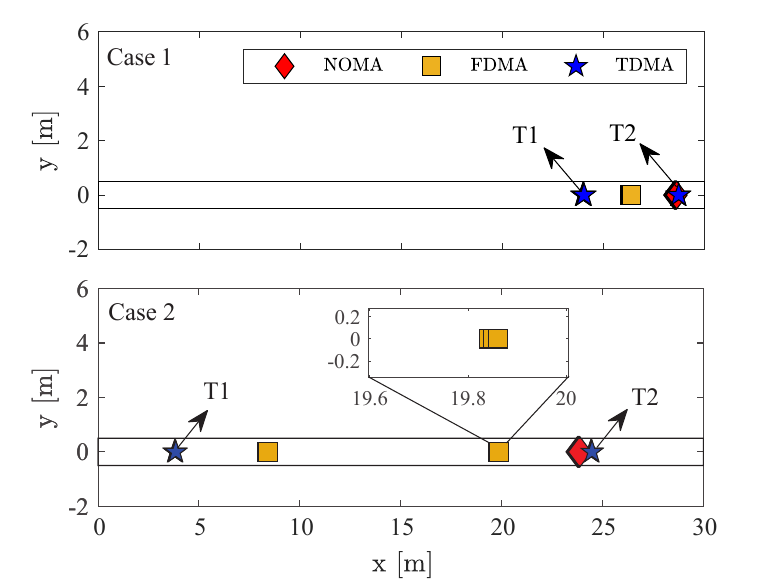}
		\caption{PA deployments for different MA schemes.}
		\label{Fig:Fig_Distribution}
	\end{minipage}
\end{figure*}

\section{Simulation Results}\label{Sec_Simulation}
In this section, numerical results are provided to validate the effectiveness of the proposed two-stage algorithm, and explore the most suitable MA schemes for PASS-based multi-user communications. Unless specified otherwise, the simulation parameters used in this paper are as follows: the height of the waveguide is set to $d=3$ m, the side length of the square area is set to $L=15$ m, the number of PAs is set to $N=6$, the target rate of the user is set to $\gamma_k= 3$ bps/Hz, the carrier frequency is set to $f_c = 28$ GHz, the user noise power is set to $\sigma_k^2 = -90$ dBm, the minimum spacing between PAs is set to $\delta = \frac{\lambda}{2}$, and $n_\text{neff} = 1.4$. For comparison, we include two baselines: 1) conventional single antenna system, where the AP is equipped with one antenna fixed at location $[0,0,d]$ (labeled Con1.), and 2) conventional antenna array system, where the AP is equipped with the same number of antennas as PASS, but with only one RF chain active and fixed at location $[0,0,d]$, employing analog beamforming (labeled Con2.).

Fig. \ref{Fig:Fig_NvsP} shows the transmit power required by different MA schemes versus the number of PAs. It can be observed that compared to the conventional antenna baselines, the required transmit power at the AP is significantly reduced with the aid of PASS. This underscores the advantages of PASS. Meanwhile, it is quite expected that NOMA consistently outperforms FDMA with PASS. However, for TDMA, it is interesting to see that the required transmit power by TDMA is less than NOMA when $N\geq4$, which is due to the time-switching feature of PASS via TDMA. Note that although TDMA can utilize the advantage of time-switching to obtain more efficiency, it still incurs extra deployment complexity and overhead. In contrast, without PASS, NOMA always requires less transmit power than both TDMA and FDMA. 

Fig. \ref{Fig:Fig_RatevsP} illustrates the transmit power at the AP versus the target rate $\gamma_1$ of user 1, with the target rate $\gamma_2$ of the user 2 fixed at $3$ bps/Hz. First, we can observe that TDMA requires less transmit power than NOMA when the target rates of the two users are symmetric, i.e., $\gamma_1=\gamma_2=3$ bps/Hz. However, when the target rates of the two users become asymmetric, the required transmit power by NOMA is less than that of OMA. This can be explained by the fact that NOMA has higher spectrum efficiency than OMA under user channels with disparity, even with TDMA scheme deployed to utilize the advantage of time-switching feature. 

Fig. \ref{Fig:Fig_Distribution} shows the optimized PA activation location for different MA schemes under two user distribution scenarios, namely Case 1 and Case 2, with the area size is $L = 30$ m. For Case 1, two users are closely located at $\psi_{\mathrm{u}}^1=\left[24.01, -10.74, 0\right]$ and $\psi_{\mathrm{u}}^2=\left[28.72, -0.44, 0\right]$. The optimized PA activation locations along x-axis are approximately $28.57$ m for NOMA and $26.36$ m for FDMA. For TDMA, leveraging the time-switching feature, the PA cluster is located at $28.71$ m for user 1 (see T1) and $24.01$ m for user 2 (see T2). Note that the PAs for NOMA are deployed closer to the user who is closer to the waveguide. Thus the channel conditions of users become more asymmetric, which is preferable for NOMA. Instead, FDMA places PAs more symmetrically to balance the channel gains of both users. For Case 2, with users widely separated at $\psi_{\mathrm{u}}^1=\left[24.44, 12.17, 0\right]$ and $\psi_{\mathrm{u}}^2=\left[3.81, 12.40, 0\right]$, the optimized PA activation location for NOMA is around $23.81$ m. FDMA splits the PAs into two symmetric clusters centered at $19.85$ m and $8.40$ m, respectively. For TDMA, PA locations are around $24.44$ and $3.80$ m to serve each user, respectively. For NOMA, it is preferable to deploy the PAs in an asymmetric scheme to achieve different channel conditions for two users. By contrast, FDMA benefits from separating the PA clusters to serve users individually, enhancing received signal strength more effectively than a shared central deployment. These results in Fig. \ref{Fig:Fig_Distribution} demonstrate important guidelines for PASS deployment for different MA schemes.
 
\section{Conclusion}\label{Sec_Conclusion}
This paper investigated a PASS-based two-user communication system for NOMA, FDMA, and TDMA schemes. For each MA scheme, a pinching beamforming optimization problem was formulated to minimize the transmit power at AP subject to the user rate requirements. To solve the resulting non-convex problems, a two-stage algorithm was developed for NOMA and FDMA to obtain desired solutions, and an optimal pinching beamforming was derived for TDMA with the aid of time-switching feature of PASS. Our numerical results revealed that PASS are capable of achieving a significant performance gain over conventional antenna systems. Furthermore, our results also showed that, NOMA always has superior performance than FDMA, while TDMA outperforms NOMA for symmetric user rate requirements by employing time-switching feature of PASS.
\appendices

\bibliographystyle{IEEEtran}

\bibliography{reference}
\newpage

\vfill

\end{document}